\documentclass[11pt,a4paper]{article}

\usepackage[utf8]{inputenc}
\usepackage[T1]{fontenc}
\usepackage{lmodern}
\usepackage[margin=2.5cm]{geometry}
\usepackage{amsmath,amssymb,amsthm}
\usepackage{algorithm}
\usepackage{algpseudocode}
\usepackage{graphicx}
\usepackage{booktabs}
\usepackage{hyperref}
\usepackage{xcolor}
\usepackage{natbib}
\usepackage{microtype}
\usepackage{parskip}
\usepackage{comment}

\hypersetup{
  colorlinks=true,
  linkcolor=blue!60!black,
  citecolor=blue!60!black,
  urlcolor=blue!60!black
}

\newtheorem{proposition}{Proposition}

\title{%
  \textbf{IOAH3: Importance-Driven Adaptive Spatial Partitioning}\\
  \large via Graph-Cut Optimisation over H3 Hierarchical Grids
}

\author{
  Ehsaneddin Jalilian\\
  \small GeoSocial Artificial Intelligence, Interdisciplinary Transformation University Austria\\
  \small ehsaneddin.jalilian@it-u.at
}

\date{}
\begin{document}
\maketitle

\begin{abstract}
We present IOAH3 (Importance-Oriented Adaptive H3 partitioning), a
computational method for constructing data-driven spatial partitions of
geo-referenced observation domains. Standard approaches to spatial aggregation
adopt fixed areal units, such as administrative boundaries or uniform
hexagonal grids at a single resolution, without regard to the informational
content of the underlying observations in each region. This leads to the
well-known modifiable areal unit problem: statistical and inferential results
depend on the arbitrary choice of partition, and spatially concentrated
phenomena are averaged out in coarse cells that obscure fine-scale structure.
IOAH3 addresses this by constructing an adaptive partition in three stages:
multi-source feature extraction and importance scoring via principal component
analysis over road density, POI density, building density, and terrain
roughness signals, with population and flood-hazard data entering as auxiliary
inputs to cell filtering and spatial smoothness; spatial
cell selection via Markov Random Field graph-cut optimisation, which jointly
maximises per-cell importance while enforcing spatial contiguity; and
data-driven hierarchical refinement of high-importance regions to finer H3
resolution levels, with neighbour-propagated support to avoid isolated
fine-resolution islands. The resulting partitions serve as input to spatial inference pipelines and provide a principled resolution of the partition-sensitivity problem prior to any modelling step.
Code is available at \url{https://github.com/EhsaneddinJalilian/IoaH3}.
\end{abstract}

\newpage

\section{Introduction}
\label{sec:intro}

Spatial inference systems that operate on geo-referenced data must first
aggregate raw observations into discrete areal units before any modelling step
can be applied. The choice of those units is consequential: coarse cells
suppress within-cell variation, fine cells may be data-sparse, and neither
uniform grid nor administrative boundary systems are designed to align with the
spatial structure of the phenomenon under study. This sensitivity, the
modifiable areal unit problem (MAUP) \citep{openshaw1984modifiable}, is
well-documented in the spatial statistics literature, yet in practice most
GeoAI pipelines adopt a fixed resolution without examination.

The problem is not merely statistical. When the areal partition is too coarse
relative to the scale at which a spatial phenomenon varies, the aggregation
operator irreversibly destroys signal. No downstream model, however
sophisticated, can recover variation that was never present in its input. This
observation motivates a prior step: constructing the partition itself in a
principled, data-driven way that aligns cell boundaries with the spatial
structure of the observations.

IOAH3 is a method for this prior step. It takes multi-source raster and vector
observations over a geographic domain and produces an adaptive H3 hexagonal
partition whose resolution varies spatially according to the estimated
importance of each region. The method combines three independently motivated
components: unsupervised importance scoring via PCA over heterogeneous
feature signals, spatial cell selection via graph-cut optimisation over a
Markov Random Field (MRF) encoding both per-cell importance and spatial
contiguity, and data-driven hierarchical refinement of selected cells to
finer H3 resolution levels with neighbour-propagated support.

The remainder of this document is structured as follows. Section~\ref{sec:background}
provides background on H3 hierarchical grids and MRF graph-cut optimisation.
Section~\ref{sec:method} describes the three-stage IOAH3 pipeline in detail.
Section~\ref{sec:implementation} gives implementation and Computational Complexity. Section~\ref{sec:discussion} discusses the method and limitations.

\section{Background}
\label{sec:background}

\subsection{H3 Hierarchical Hexagonal Grids}

H3 \citep{brodsky2018h3} is a hierarchical discrete global grid system that
tessellates the Earth's surface using hexagonal cells at 16 resolution levels.
At resolution $r$, each cell covers approximately $A_r$~km$^2$, where $A_r$
decreases by a factor of roughly 7 with each increment in $r$. Key resolution
levels relevant to this work are given in Table~\ref{tab:h3res}.

\begin{table}[h]
\centering
\caption{H3 resolution levels used in IOAH3.}
\vspace{2mm}
\setlength{\tabcolsep}{10pt}
\renewcommand{\arraystretch}{1.3}
\label{tab:h3res}
\begin{tabular}{cccc}
\toprule
Resolution & Avg.\ cell area (km$^2$) & Avg.\ edge length (km) & Typical use \\
\midrule
7  & $\sim$5.16   & $\sim$1.22  & Base / coarse aggregation \\
9  & $\sim$0.105  & $\sim$0.174 & Intermediate refinement \\
10 & $\sim$0.015  & $\sim$0.066 & Fine-grained urban detail \\
\bottomrule
\end{tabular}
\end{table}

H3 cells have the property that each cell at resolution $r$ is exactly
contained within one cell at resolution $r-1$, forming a strict hierarchy.
This makes hierarchical refinement, replacing a coarse cell with its seven
child cells at the next resolution, geometrically exact and computationally
inexpensive.

\subsection{Markov Random Fields and Graph-Cut Optimisation}

A Markov Random Field (MRF) over a set of sites $\mathcal{V}$ with
neighbourhood structure $\mathcal{E}$ is a probability model of the form
\begin{equation}
  P(\mathbf{x}) \propto \exp\!\left(
    -\sum_{i \in \mathcal{V}} \psi_i(x_i)
    - \lambda \sum_{(i,j) \in \mathcal{E}} \phi_{ij}(x_i, x_j)
  \right),
\label{eq:mrf}
\end{equation}
where $\psi_i$ are unary potentials encoding per-site preferences and
$\phi_{ij}$ are pairwise potentials encoding smoothness over adjacent sites.
For binary labellings $x_i \in \{0, 1\}$ (included / excluded) and
submodular pairwise terms, the MAP estimate of \eqref{eq:mrf} can be computed
exactly in polynomial time via the max-flow / min-cut duality
\citep{boykov2001fast, kolmogorov2004energy}. The resulting cut minimises
a weighted combination of unary dissatisfaction and pairwise discontinuity,
producing a labelling that is simultaneously consistent with per-site
importance and spatially coherent.

\section{Method}
\label{sec:method}

\subsection{Stage 1: Feature Extraction and Importance Scoring}
\label{sec:stage1}

Let $\mathcal{H}_r$ denote the set of H3 cells covering the domain $\Omega$ at
base resolution $r$. For each cell $h \in \mathcal{H}_r$, IOAH3 extracts a
feature vector $\mathbf{f}(h) \in \mathbb{R}^d$ aggregating the following
information sources:

\begin{itemize}
  \item \textbf{Road density} $\rho_{\text{road}}(h)$: count of OSM road
    segments of class motorway, trunk, primary, secondary, tertiary,
    unclassified, or residential whose nodes fall within $h$, normalised by
    cell area.
  \item \textbf{POI density} $\rho_{\text{poi}}(h)$: count of OSM point-of-interest
    nodes (amenity or shop tags) within $h$, normalised by cell area.
  \item \textbf{Building density} $\rho_{\text{bldg}}(h)$: count of OSM building
    footprint nodes within $h$, normalised by cell area.
  \item \textbf{Elevation roughness} $\sigma_z(h)$: mean gradient magnitude
    within $h$, computed at the native DEM resolution (10\,m Lambert projection)
    as $\|\nabla z\| = \sqrt{(\partial z/\partial x)^2 + (\partial z/\partial y)^2}$. The value then is averaged over sampled pixels mapped to each
    H3 cell. This captures terrain ruggedness, a stronger predictor of
    infrastructural and hazard importance than mean elevation. 
    \vspace{2mm}
    
  \item \textbf{Population count} $p(h)$: total resident population within $h$,
    aggregated. Population is not included
    in the PCA feature set. It enters the pipeline as: (i)~a pre-filtering criterion
    (cells with zero population, no infrastructure, and negligible hazard are
    dropped before PCA); and (ii)~an input to the MRF pairwise smoothness weight
    alongside flood-hazard probability (see Section~\ref{sec:stage2}).
  \item \textbf{Flood-hazard probability} $\eta(h)$: mean flood-hazard probability
    within $h$, aggregated by mean from a flood-probability raster. Like
    population, hazard does not enter the PCA feature set but serves as a
    pre-filtering criterion and as a pairwise smoothness input in Stage~2.
\end{itemize}

The first four signals (road density, POI density, building density, elevation
roughness; $d = 4$) form the PCA feature matrix. They are standardised to zero
mean and unit variance, then reduced to a single importance score via principal
component analysis:
\begin{equation}
  s(h) = \mathrm{normalise}\!\left(\mathbf{w}^\top \tilde{\mathbf{f}}(h)\right)
         \in [0, 1],
\label{eq:pca_score}
\end{equation}
where $\mathbf{w}$ is the first principal component loading vector and
$\tilde{\mathbf{f}}(h)$ is the standardised feature vector. The normalisation
maps $s(h)$ to $[0,1]$ by min-max scaling across all cells. This score is
data-driven and parameter-free: no manual weighting of features is required.

\subsection{Stage 2: Graph-Cut Cell Selection}
\label{sec:stage2}

Given the importance scores $\{s(h)\}_{h \in \mathcal{H}_r}$, IOAH3
constructs an MRF over the H3 grid and solves for the binary labelling
$\ell: \mathcal{H}_r \to \{0, 1\}$ (included / excluded) that minimises

\begin{equation}
  E(\ell) = \sum_{h \in \mathcal{H}_r} \psi_h(\ell_h)
           + \lambda \sum_{(h,h') \in \mathcal{E}} \phi_{hh'}(\ell_h, \ell_{h'}),
\label{eq:energy}
\end{equation}

where the unary term $\psi_h(0) = 1 - s(h)$ (cost of including a low-importance
    cell) and $\psi_h(1) = 0$ (no cost for excluding any cell). The pairwise term $\phi_{hh'}(\ell_h, \ell_{h'}) = \mathbf{1}[\ell_h \neq \ell_{h'}]
    \cdot w_{hh'}$, where the smoothness weight is
    \begin{equation}
      w_{hh'} = e^{-\left(|p(h) - p(h')| + |\eta(h) - \eta(h')|\right)},
    \label{eq:smoothness}
    \end{equation}
    with $p(h)$ the raw population count and $\eta(h)$ the raw
    hazard value of cell $h$. Raw (un-normalised) values are used directly;
    the exponential suppresses the weight when adjacent cells differ strongly
    in population or hazard, regardless of absolute scale. $\mathcal{E}$ is the H3 k-ring-1 adjacency graph (each cell connected
    to its six hexagonal neighbours). $\lambda = 1.0$ controls the smoothness-importance trade-off.

The energy \eqref{eq:energy} is submodular in the pairwise terms (since
$w_{hh'} \geq 0$) and is therefore minimised exactly by max-flow / min-cut
using \citep{boykov2001fast}.

\vspace{4mm}

\begin{proposition}
The pairwise term \eqref{eq:smoothness} is submodular for all
$w_{hh'} \geq 0$, and the total energy \eqref{eq:energy} admits an exact
minimum-cut solution.
\end{proposition}

\begin{proof}

Submodularity of the term $\phi_{hh'}$ follows directly from
$w_{hh'} \geq 0$ and the standard Potts model construction. The result then follows from \citet{kolmogorov2004energy}, Theorem~1. 
\end{proof}

After optimisation, cells with $\ell_h = 0$ are marked \emph{included};
cells with $\ell_h = 1$ are marked \emph{excluded} and retained as
background context. An additional connectivity filter removes included cells
with fewer than one included neighbour, preventing isolated singleton
inclusions.

\subsection{Stage 3: Adaptive Hierarchical Refinement}
\label{sec:stage3}

Included cells are refined to finer H3 resolutions based on their importance
score. Let $q_{40}, q_{65}, q_{80}$ denote the 40th, 65th, and 80th
percentiles of $\{s(h) : \ell_h = 0\}$. The target resolution for cell $h$ is

\begin{equation}
  r^*(h) = \begin{cases}
    r_{\max} & \text{if } s(h) > q_{80}, \\
    9        & \text{if } q_{65} < s(h) \leq q_{80}, \\
    r_0      & \text{otherwise},
  \end{cases}
\label{eq:target_res}
\end{equation}

where $r_0 = 7$ is the base resolution and $r_{\max} = 10$ is the maximum
refinement resolution.

Critically, refinement is spatially propagated: for every cell $h$ with
$r^*(h) > r_0$, all H3 k-ring-1 neighbours of $h$ are also promoted to at
least $r^*(h)$. This ensures contiguous fine-resolution zones rather than
isolated fine cells surrounded by coarse cells, which would create
discontinuous partition boundaries.

Each cell $h$ with target resolution $r^* > r_0$ is replaced by its
$7^{r^*-r_0}$ H3 children at resolution $r^*$. Feature values for child
cells are sampled from fine-resolution rasters pre-loaded once before the
refinement loop.  Pseudocode for the full pipeline is given in Algorithm~\ref{alg:ioah3}.

\begin{algorithm}[h]
\caption{IOAH3: Importance-Driven Adaptive H3 Partitioning}
\label{alg:ioah3}
\begin{algorithmic}[1]
\Require Raster layers $\mathcal{R}$, OSM vector data $\mathcal{V}$,
         base resolution $r_0$, max resolution $r_{\max}$
\Ensure  Adaptive partition $\mathcal{P}$ with per-cell resolution and features

\State Extract per-cell features $\mathbf{f}(h)$ for all $h \in \mathcal{H}_{r_0}$
  from $\mathcal{R}$ and $\mathcal{V}$
\State Compute importance scores $s(h) \leftarrow \mathrm{PCA}_1(\mathbf{f}(h))$
  via \eqref{eq:pca_score}
\State Construct MRF graph $G = (\mathcal{H}_{r_0}, \mathcal{E})$ with
  unary and pairwise terms per \eqref{eq:energy}--\eqref{eq:smoothness}
\State Solve $\ell^* \leftarrow \arg\min_\ell E(\ell)$ via max-flow / min-cut
\State \textbf{Connectivity filter:} for each $h$ with $\ell^*_h = 0$,
  set $\ell^*_h \leftarrow 1$ if $|\{h' \in \mathrm{kring1}(h) : \ell^*_{h'} = 0\}| < 1$
\State Compute quantiles $q_{40}, q_{65}, q_{80}$ over $\{s(h) : \ell^*_h = 0\}$
\State Compute target resolutions $r^*(h)$ per \eqref{eq:target_res}
\State Propagate refinement to neighbours:
  \For{$h$ with $r^*(h) > r_0$}
    \For{$h'$ in k-ring-1 neighbours of $h$}
      \State $r^*(h') \leftarrow \max(r^*(h'), r^*(h))$
    \EndFor
  \EndFor
\State \textbf{Refine:}
  \For{$h$ with $r^*(h) > r_0$}
    \State Replace $h$ with children $\{c : c \in \mathrm{H3children}(h, r^*(h))\}$
    \State Sample feature values for each child from fine rasters
  \EndFor
\State \Return partition $\mathcal{P}$ with per-cell resolution, features,
  and inclusion label
\end{algorithmic}
\end{algorithm}

\section{Implementation and Computational Complexity}
\label{sec:implementation}

Let $n = |\mathcal{H}_{r_0}|$ be the number of base-resolution cells. Stage~1
runs in $O(nd)$ time for feature extraction and $O(nd)$ for PCA scoring.
Stage~2 constructs an MRF with $n$ nodes and at most $6n$ edges (H3 hexagonal
adjacency), and solves it in $O(n^2 \sqrt{n})$ worst-case time via the
Boykov-Kolmogorov max-flow algorithm, though empirical runtimes are
substantially faster on geographic grids. Stage~3 runs in $O(n \cdot
7^{r_{\max} - r_0})$ in the worst case (all cells refined to maximum
resolution), but in practice only a small fraction of cells are refined to
$r_{\max}$.

\section{Discussion and Limitations}
\label{sec:discussion}

IOAH3 addresses the MAUP sensitivity problem at the partition construction
stage, prior to any modelling step. The graph-cut formulation is principled in
that it admits an exact solution and makes the importance-smoothness trade-off
explicit through the $\lambda$ parameter. The hierarchical H3 structure
guarantees that refinement is geometrically consistent and that child-cell
feature values can be sampled directly from the same raster sources.

Several limitations are worth acknowledging. First, the PCA importance score
is a linear aggregation of the input features; non-linear importance structure
(for example, interaction between road density and hazard level) is not
captured. A learned importance function, trained on downstream task
performance, would be more expressive but requires labelled data. Second, the
MRF smoothness term is defined over population and hazard signals only; a more
general formulation would include all feature signals in the pairwise weights.
Third, the quantile-based refinement thresholds in \eqref{eq:target_res} are
fixed at the 40th, 65th, and 80th percentiles; these were chosen empirically
and may require adjustment for domains with different importance distributions.

\section{Conclusion}
\label{sec:conclusion}

We have presented IOAH3, a three-stage method for constructing adaptive spatial
partitions of geo-referenced observation domains over H3 hierarchical grids.
The method combines PCA-based importance scoring, MRF graph-cut optimisation
for spatially coherent cell selection, and data-driven hierarchical refinement
with neighbour propagation. The resulting multi-resolution partitions
concentrate fine-grained cells in regions of high informational importance
while maintaining a coarser but complete background partition. The
method is fully automated, parameter-light, and produces citable, reproducible
output suitable as input to downstream spatial inference pipelines.

\bibliographystyle{plainnat}
\bibliography{ioah3_refs}

\end{document}